\documentclass[11pt,article]{amsart}
\usepackage{amssymb}
\usepackage{amsfonts}
\usepackage{amsbsy,url}
\usepackage{latexsym}
\usepackage{amssymb,latexsym,amsmath,amsthm}
\usepackage{color} 
\usepackage{graphicx}
\usepackage{subcaption}
\theoremstyle{plain}
\newcounter{Theorem}
\newtheorem{theorem}[Theorem]{Theorem}
\newtheorem{lemma}[Theorem]{Lemma} \newtheorem{proposition}[Theorem]{Proposition}
\newtheorem{corollary}[Theorem]{Corollary} \theoremstyle{remark} \newcounter{Definition}

\newtheorem{definition}[Definition]{Definition}

\newcommand{\diag}{\operatorname{diag}}

\newtheorem{example} {\bf Example} [section]

\numberwithin{equation}{section}
\makeatletter
\@namedef{subjclassname@2020}{%
  \textup{2020} Mathematics Subject Classification}
\makeatother

\begin{document}
\title{Fidelity preservation for mixed unitary quantum channels}

\keywords{}

\maketitle

\begin{center}
Kai Liu\\
Department of Mathematics, University of Central Florida, Orlando, FL 32816, USA\\
Email: kai.liu@ucf.edu\\[4pt]

Deguang Han\\
Department of Mathematics, University of Central Florida, Orlando, FL 32816, USA\\
Email: deguang.han@ucf.edu\\[4pt]

\textbf{Corresponding author:} Kai Liu (Email: kai.liu@ucf.edu, Phone:+1 4079252158)
\end{center}

\begin{abstract}

Distinguishable and non-distinguishable quantum states play different roles in quantum information processing, and environmental noise generally alters their relative geometry. In this paper, we study exact preservation of the Uhlmann fidelity for selected pairs of pure states under mixed unitary quantum channels of the form $ \Phi(\rho)=\sum_{i=1}^N p_i U_i \rho U_i^*. $ Rather than characterizing full state protection, quantum error correction, or decoherence-free subspaces, our aim is to determine when the weaker quantity \(F(\rho_1,\rho_2)\) remains unchanged by the channel itself. For distinguishable pairs in qubit systems, we obtain an algebraic characterization in terms of the products \(U_j^*U_i\). For non-distinguishable pure pairs, we derive a preservation criterion for two-unitary mixtures in terms of a symmetry relation between the states and the relative unitary \(U_1^*U_2\). We further study general phase damping channels as a case study, showing that exact fidelity preservation is highly constrained even though decoherence generically increases fidelity through contraction of off-diagonal terms. These results show that fidelity preservation is a weaker state-dependent notion of robustness, complementary to full error correction and noiseless protection.

\end{abstract}

\noindent\textbf{Keywords:} Quantum channels, Quantum property, Fidelity preservation, Distinguishable and Non-distinguishable states

\maketitle

\newpage

\section{Introduction}

\noindent The study of quantum channels and their action on quantum states is central to quantum information theory and its applications, with direct implications for emerging quantum technologies. In communication and computation, reliable performance requires relevant quantum resources to survive under the action of noise \cite{Preskill,Gisin,Sidhu}. Distinguishable states are essential for encoding and transmitting classical information reliably, while non-distinguishable states are closely related to quantum coherence, superposition, and other genuinely quantum effects \cite{Fuchs,Monta,Wang,Benatti,Streltsov,Marvian}. Under environmental noise, these features are generally altered, and understanding which aspects of a quantum system survive the action of a channel is a basic problem in the study of open quantum systems and quantum information processing.\\
\\
A classical approach for combating noise is \emph{quantum error correction} (QEC) \cite{Knill2000,Lidar2013}. For a noise process $\Phi$ and quantum state $\rho$, one seeks a recovery map $\mathcal{R}$ such that $(\mathcal{R}\circ\Phi)(\rho)=\rho$ for every state in a chosen code subspace. The Knill--Laflamme conditions characterize when such exact correction is possible in terms of algebraic relations among the Kraus operators of the channel restricted to the code \cite{Knill2000}. Closely related are \emph{decoherence-free subspaces} and \emph{noiseless subsystems}, where the noise acts trivially on a protected subspace or subsystem and hence preserves the full state information. These frameworks are concerned with preserving the full quantum state, either exactly or through recovery.\\
\\
In this paper, we adopt the complementary viewpoint of \emph{quantum property preservation} \cite{Kumar2025}. Rather than asking whether the entire state can be restored, we focus on whether a specific property of a pair of states remains unchanged under the channel itself. The property considered here is the Uhlmann fidelity
\[
F(\rho_1,\rho_2),
\]
and we ask when it is preserved by a mixed unitary channel of the form
\[
\Phi(\rho)=\sum_{i=1}^N p_i U_i\rho U_i^*,
\]
where the $U_i$ are unitary, $p_i>0$, and $\sum_i p_i=1$. Thus we study the equality
\begin{equation}\label{eq:F-preserve}
F(\rho_1,\rho_2)=F(\Phi(\rho_1),\Phi(\rho_2)).
\end{equation}
Our goal is not to characterize full state protection, quantum error correction, or decoherence-free subspaces, but rather to determine when the weaker quantity $F(\rho_1,\rho_2)$ is preserved by the channel itself. This weaker notion is still nontrivial: fidelity preservation for a selected pair may occur even when superpositions on its span lose coherence and no nontrivial correctable code exists. Throughout, we track the action of $\Phi$ on quantum states by purifications $|\Psi\rangle$ associated with pure inputs $|\varphi\rangle$, since this viewpoint helps expose structural constraints.\\
\\
We treat two regimes separately. For distinguishable pure-state pairs, we study when orthogonality, equivalently zero fidelity, is preserved under a mixed unitary channel. In the qubit case, this leads to an exact characterization in terms of the diagonal structure of the relative products $U_j^*U_i$, and we then extend this perspective to several higher-dimensional settings. For non-distinguishable pure-state pairs, the problem becomes an equality case of the monotonicity of fidelity, and we derive a criterion for two-unitary mixtures in terms of a symmetry relation involving the relative unitary $U_1^*U_2$. We also study a representative and practically significant class of channels, namely \emph{general phase damping} (GPD) channels, which provide a transparent setting in which exact fidelity preservation can be compared with generic decoherence.\\
\\
\medskip\noindent
\textbf{Organization:}\\
Section 2 recalls the basic notions, classic ideas, and results that will be used in our work. Section 3 studies fidelity preservation for distinguishable pure states and compares the resulting conditions with the Knill--Laflamme criterion from quantum error correction. Section 4 turns to non-distinguishable pure states and obtains a preservation criterion for two-unitary mixtures. Section 5 analyzes general phase damping channels as a concrete model in which exact fidelity preservation and decoherence can be studied simultaneously. We conclude with a brief summary and outlook. The full proofs are deferred to the appendices.

\section{Preliminaries}
\noindent Let $H$ be the quantum system. A \textit{pure state} $|\varphi\rangle$ from $H$ refers to the unit vector in H, which is also described by the associated density operator $\rho=|\varphi\rangle\langle\varphi|$. The \textit{mixed state} is a statistical ensemble of pure states $\{|\varphi_i\rangle\}$ given by $\rho=\sum_{i}p_i|\varphi_i\rangle\langle\varphi_i|$ with $\sum_{i}p_i=1$.  Fidelity is a standard measure of the closeness of the quantum states. In general, for a pair of quantum states with associated density operators $\rho_1, \rho_2$,  the \textit{fidelity} between them (denoted by $F(\rho_1, \rho_2)$) is defined by $F(\rho_1, \rho_2)=\text{Tr}\sqrt{\sqrt{\rho_2}\rho_1\sqrt{\rho_2}}$ or $F(\rho_1, \rho_2)=||\sqrt{\rho_1}\sqrt{\rho_2}||_1$ by using the trace norm $||\cdot||_1$. More restrictively, for pure states $|\varphi_1\rangle$, $|\varphi_2\rangle$ and the associated density operators $\rho_1$, $\rho_2$, the fidelity between them can be directly defined as $F(\rho_1, \rho_2)=|\langle \varphi_1, \varphi_2\rangle|$ \cite{watrous2018,barman2000,jozsa1994}.\\
\\
A quantum channel refers to a generalization of classical communication in a quantum system that describes the way in which quantum states are transmitted. Mathematically, a \textit{quantum channel} is usually represented by a completely positive trace-preserving (CPTP) linear map $\Phi: B(H)\to B(K)$ between quantum systems $H$ and $K$.  The Kraus theorem shows that, a CPTP map $\Phi$ from $B(H)$ to $B(K)$ with Choi's rank $r$ (denoted by $\text{Cr}(\Phi)=r$) always has a simple Kraus representation:
$
\Phi(\rho) = \sum_{i=1}^{r}A_i\rho A_i^*$, for any state $ \rho$
and some operators $A_1, . . . , A_r\in B(H, K)$ satisfying $\sum_{i=1}^rA_i^*A_i=I$. We say that a quantum channel $\Phi$ is mixed unitary if $\Phi(\rho)=\sum_{i=1}^N p_iU_i\rho U_i^*$ for some $U_i\in U(H)$ and $\sum_{i=1}^Np_i=1$. Further, a quantum channel $\Phi$ is said to be \textit{degenerated} if $\text{Cr}(\Phi)=1$ and \textit{non-degenerated}, otherwise \cite{watrous2018,landau1993birkhoff,cozzini2007,viamon2009,chir2008}.\\
\\
 For a pair of quantum states $\rho_1$ and $\rho_2$, we say that they are \textit{distinguishable} if $F(\rho_1, \rho_2)=0$ and \textit{non-distinguishable} if $F(\rho_1, \rho_2)\ne 0$. For pure states $|\varphi_1\rangle$ and $|\varphi_2\rangle$, it is clear that they are distinguishable if $\langle \varphi_1, \varphi_2\rangle=0$ and non-distinguishable, otherwise. Furthermore, for a given quantum channel $\Phi$, we say $\Phi$ preserves the fidelity between quantum states $\rho_1, \rho_2$ if \eqref{eq:F-preserve} is true. For $S$, a set of quantum states, it is preserved by $\Phi$ if \eqref{eq:F-preserve} is always true for any pair of states from $S$. Throughout this paper, we restrict attention to distinct pure input states and mixed unitary quantum channels, since our aim is to obtain explicit structural criteria for exact fidelity preservation in this setting.\\
 \\
 \noindent Let $H$ be an $n-$dimensional quantum system and $\rho$ be a quantum state from $H$. A \emph{purification} of $\rho$ is a pure state $|\psi\rangle \in \mathcal{H} \otimes \mathcal{K}$ for some auxiliary Hilbert space $\mathcal{K}$ such that $\rho = \operatorname{Tr}_{\mathcal{K}} \big( |\psi\rangle \langle \psi| \big),$ where $\operatorname{Tr}_{\mathcal{Z}}$ denotes the partial trace over $\mathcal{Z}$. For the output of $\rho$ under a mixed unitary quantum channel  $\Phi(\rho)=\sum_{i=1}^Np_iU_i\rho U_i^*$. A standard purification $|\Psi\rangle$ would be 
\begin{equation}\label{2.2}
	|\Psi\rangle=\sum_{k}\sqrt{p_k}(U_k|\varphi\rangle\otimes |k\rangle),
\end{equation}
where $\{|k\rangle\}, 1\leq k\leq N$, is an orthonormal basis of a $N-$dimensional ancilla space $\mathcal{Z}=\mathbb{C}^N$.\\

\noindent Furthermore, the image of each pure state $\rho$ under the mixed unitary channel $\Phi$ can be reconstructed from the above purification. More precisely, for $1\leq k, k_1, k_2\leq N$, we have
\begin{equation}
	\begin{aligned}
		\Phi(\rho)&=\sum_{k}p_kU_k\rho_i U_k^*\\
		&=\text{Tr}_Z\sum_{k_1, k_2}\sqrt{p_{k_1}}\sqrt{p_{k_2}}U_{k_1}|\varphi\rangle\langle \varphi|U_{k_2}^*\otimes |k_1\rangle \langle k_2|\\
		&=\text{Tr}_Z\left(\sum_k\sqrt{p_k}U_i|\varphi\rangle\otimes |k\rangle\right) \text{Tr}_Z\left(\sum_k\sqrt{p_k}\langle \varphi|U_k^*\otimes \langle k|\right)\\
		&=\text{Tr}_Z\left(|\Psi\rangle\langle\Psi|\right).
	\end{aligned}
\end{equation}
For each pair of pure states $|\varphi_1\rangle, |\varphi_2\rangle$ with associated density operators $\rho_1, \rho_2$, we have 
\begin{equation}
	F(\Phi(\rho_1), \Phi(\rho_2))= F\left(Tr_Z(|\Psi_1\rangle\langle\Psi_1|\right ), Tr_Z\left(|\Psi_2\rangle\langle\Psi_2|)\right ).
\end{equation}

\begin{lemma}\label{nuclear norm}
	Let $|\Psi_1\rangle, |\Psi_2\rangle$ be purifications of quantum states in the quantum system $H\otimes \mathcal{Z}$. It holds that
	\begin{equation}
		F\left(\text{Tr}_Z(|\Psi_1\rangle\langle\Psi_1|), \text{Tr}_Z(|\Psi_2\rangle\langle\Psi_2|)\right )=||\text{Tr}_H(|\Psi_2\rangle\langle\Psi_1|)||_1.
	\end{equation}  
\end{lemma}
\noindent The mapping $\text{Tr}_H: D(Z)\rightarrow D(Z)$ is known as the \textit{complementary channel} of $\Phi$ and has the following representation if the Kraus representation of $\Phi$ is specified.
\begin{equation}
	\Psi(\rho)=\sum_{j,k=1}^r\sqrt{p_jp_k}\langle U_k^*U_j, \rho\rangle E_{jk}.
\end{equation}
\\

 \noindent The following two theorems regarding quantum channels are essential to our discussion.
\begin{theorem} \label{Choi-thm2} \cite{gupta2015} Suppose that  $\Phi: B(H)\rightarrow B(K)$ is a quantum channel with Kraus representation $\Phi(\rho) = \sum_{j=1}^{r}A_{j}\rho A_{j}^{*} $  for some $A_{j}\in B(H,K)$, $1\leq j\leq r$. Then $\Phi(\rho) = \sum_{j=1}^{m}B_{j}\rho B_{j}^{*}$ for operators $B_{j}\in B(H,K)$, $1\leq j\leq m$ if and only if  there exists $U = (u_{ij})$ such that $U^{*}U  = I_{r}$ and $A_{i}  = \sum_{j=1}^{r}u_{ij}B_{j}$ for $i=1,... , m$.
\end{theorem}
\begin{theorem}[Dilation theorem]\label{Steinspring theorem}\cite{gupta2015}
	Let $H, K, Z$ be complex Hilbert spaces. For any quantum channel $\Phi$, there exists $A\in B(H, K\otimes Z)$ such that $\Phi(\rho)=\text{Tr}_{\mathcal{Z}}A\rho A^*$ and $A^*A=I_H$ for every $\rho\in D(H)$, where $\text{Tr}_{\mathcal{Z}}$ is the partial trace according to the space $\mathcal{Z}$.
\end{theorem}

\section{Fidelity Preservation for Distinguishable Quantum States}
\noindent We first consider the distinguishable case, where $F(\rho_1,\rho_2)=0$. In this regime, the question is whether a fixed distinguishable pure-state pair remains distinguishable after the action of the channel. This is a much weaker problem than preserving an entire orthonormal family or a full quantum code. For qubit systems, the following lemma shows that it is sufficient to work with mixed unitary channels.
\begin{lemma}\label{Landau}\cite{landau1993birkhoff} Every unital qubit quantum channel $\Phi$ is a mixed unitary.	
\end{lemma}
\noindent Since distinguishability of pure states is equivalent to zero fidelity, the preservation problem in this section is precisely the question of when the outputs of the channel still satisfy
\begin{equation}
	F\bigl(\Phi(|\phi_1\rangle\langle\phi_1|),\Phi(|\phi_2\rangle\langle\phi_2|)\bigr)=0.
\end{equation}
The following theorem gives an exact characterization in the qubit case.
\begin{theorem}\label{2d}
	Suppose that $|\varphi_1\rangle, |\varphi_2\rangle$ are distinguishable states in the Qubit system $H$. Let $\Phi(\rho)=\sum_{i=1}^N p_i U_i\rho U_i^*$ be a mixed unitary quantum channel over $H$. Then the following are equivalent:
\\	
$(i)$ $\Phi$ preserves distinguishable states $|\varphi_1\rangle$, $|\varphi_2\rangle$.\\ 
$(ii)$  There exists a basis such that $U_j^*U_i$ is diagonal for any $1\leq i, j\leq N$. \\
$(iii)$  There exists a basis such that, for any Kraus representation $\Phi(\rho)=\sum_{i=1}^MA_i\rho A_i^*$ of $\Phi$, $A_j^*A_i$ are diagonal for any $1\leq i, j\leq M$.
\end{theorem}
\noindent The proof is presented in the Appendix. The following example illustrates how such quantum channels in Theorem \ref{2d} can be resembled. Further, this example is also unitary equivalent to the general phase damping quantum channel, which will be introduced in Section 5.
\begin{example}\label{example 2.1}
	Let $\Phi(\rho)=\frac{1}{3}E_1\rho E_1^*+\frac{2}{3}E_2\rho E_2^*$ be a quantum channel over qubit system $H$ with Kraus operators
	\begin{equation}
		E_1=\begin{pmatrix}{}
  \frac{1}{\sqrt{2}}&  \frac{1}{\sqrt{2}}i\\
   \frac{1}{\sqrt{2}}i&  \frac{1}{\sqrt{2}}
\end{pmatrix},
E_2=\begin{pmatrix}{}
  \frac{1}{\sqrt{2}}&  -\frac{1}{\sqrt{2}}\\
   \frac{1}{\sqrt{2}}i&  \frac{1}{\sqrt{2}}i
\end{pmatrix}.
	\end{equation}
It is easy to verify that $E_2^*E_1$ and $E_1^*E_2$ are both diagonal
\begin{equation}
	E_2^*E_1=\begin{pmatrix}{}
   1&0  \\
   0&-i 
\end{pmatrix},
E_1^*E_2=\begin{pmatrix}{}
   1&0  \\
   0&i 
\end{pmatrix}.
\end{equation}
	And for distinguishable states $|0\rangle$, $|1\rangle$, we have
\begin{equation}
	\Phi(|0\rangle\langle 0|)=\begin{pmatrix}{}
  \frac{1}{2}&  -\frac{1}{2}i\\
   \frac{1}{2}i&  \frac{1}{2}
   \end{pmatrix},
   \Phi(|1\rangle\langle 1|)=\begin{pmatrix}{}
  \frac{1}{2}&  \frac{1}{2}i\\
   -\frac{1}{2}i&  \frac{1}{2}
   \end{pmatrix}
\end{equation}
both rank 1 density matrices satisfying $\langle \Phi(|0\rangle\langle 0|),\Phi(|1\rangle\langle 1|)\rangle=0$, which indicates the outputs are still distinguishable states.	 
\end{example}
\noindent For comparison, we recall the Knill-Laflamme condition from the quantum error correction.
\begin{theorem}[Knill--Laflamme condition]\label{Knill}\cite{Knill2000}
Let $\mathcal{C} \subseteq \mathcal{H}$ be quantum code with an orthogonal projector $P$ onto $\mathcal{C}$. 
A quantum channel $\Phi$ with Kraus operators $\{E_i\}$ is correctable on $\mathcal{C}$ 
if and only if complex numbers $\lambda_{ij}$ exist such that
\begin{equation}
P E_i^\ast E_j P = \lambda_{ij} P, \quad \forall i,j.
\end{equation}
Equivalently, if $\{| \psi_a \rangle\}$ is an orthonormal basis of $\mathcal{C}$, then
\begin{equation}
\langle \psi_a | E_i^\ast E_j | \psi_b \rangle = \lambda_{ij}\, \delta_{ab}, 
\quad \forall i,j,\, a,b.
\end{equation}
\end{theorem}
\noindent Although Theorem 5 and Theorem 6 both involve algebraic constraints on Kraus products, the two statements address different objectives. The Knill--Laflamme condition characterizes recoverability of \emph{all} states in a code subspace after the action of the channel, whereas Theorem 5 characterizes preservation of the single quantity $F(\rho_1,\rho_2)$ for one fixed distinguishable pair under the channel itself. In particular, Theorem 5 does not imply the existence of a nontrivial correctable code, nor does it imply that the span of the pair is decoherence-free or noiseless.

\begin{example}
This example shows that exact preservation of fidelity for a selected distinguishable pair does not imply noiseless protection on the span of that pair. Consider the dephasing channel on the qubit,
\begin{equation}
	\Phi(\rho)=(1-p)\rho+pZ\rho Z,\qquad
Z=
\begin{pmatrix}
1&0\\
0&-1
\end{pmatrix},
\qquad 0<p<1.
\end{equation}
For $|0\rangle$ and $|1\rangle$, we have
\begin{equation}
	\Phi(|0\rangle\langle 0|)=|0\rangle\langle 0|,
\qquad
\Phi(|1\rangle\langle 1|)=|1\rangle\langle 1|,
\end{equation}
so their fidelity is preserved. In particular, the distinguishable pair remains distinguishable under $\Phi$.\\
\\
However, the two-dimensional subspace $\mathrm{span}\{|0\rangle,|1\rangle\}$ is not noiseless. Indeed, a coherent superposition such as
\begin{equation}
	\frac{|0\rangle+|1\rangle}{\sqrt{2}}
\end{equation}
is mapped to
\begin{equation}
	\frac12\bigl(|0\rangle\langle 0|+|1\rangle\langle 1|\bigr),
\end{equation}
losing its off-diagonal coherence. Thus the channel preserves the fidelity of the selected pair $|0\rangle,|1\rangle$, but it does \emph{not} preserve all states on their span. Consequently, no nontrivial Knill--Laflamme code exists in this example, and only trivial one-dimensional codes are correctable.\\
\\
Therefore, preservation of fidelity for a chosen pair is strictly weaker than either full quantum error correction or the existence of a decoherence-free subspace.
\end{example}

\noindent We now extend the idea of Theorem~\ref{2d} to obtain the results for a \emph{qutrit system} (3-dimensional quantum system) and, more generally, for arbitrary d-dimensional quantum systems. Let $E_d=\{|i\rangle\}_{0\leq i\leq d-1}$ be the standard basis of a $d-$dimensional quantum system $H$. Matrix $U\in U(H)$ is a \textit{two-level matrix} if there exists a subspace $H_0$ and its orthogonal complement $H_0^{\perp}$ of $H$ such that the restriction of $U$ onto $H_0^{\perp}$ is the identity matrix \cite{nielson,li2013}. At this point, the case for a qutrit system is clear using the following lemma.
\begin{lemma}\label{qutrit}
	Let $H$ be a qutrit system. If  $U\in U(H)$ has symmetric off-diagonal zeros, then $U$ is a two-level unitary matrix.
\end{lemma}
\noindent See proof in Appendix and the following result holds immediately.\\

\begin{proposition}\label{3d}
	Suppose that $|\varphi_1\rangle, |\varphi_2\rangle$ are distinguishable states in qutrit system $H$. Let $\Phi(\rho)=\sum_{i=1}^N p_i U_i\rho U_i^*$ be a mixed unitary quantum channel over $H$. Then the following are equivalent:\\
$(i)$ $\Phi$ preserves distinguishable states $|\varphi_1\rangle$, $|\varphi_2\rangle$.\\ 
$(ii)$ There exists a basis such that $U_j^*U_i^*$ are two-level unitary matrices for any $1\leq i, j\leq N$. \\
$(iii)$  There exists a basis such that, for any Kraus representation $\Phi(\rho)=\sum_{i=1}^rA_i\rho A_i^*$ of  $\Phi$, $A_j^*A_i$ are two-level matrices for any $1\leq i, j\leq r$.
\end{proposition}

\begin{corollary}\label{qutrit, basis}
	Let $\Phi$ be a mixed unitary quantum channel over a qutrit quantum system $H$ such that $\Phi(\rho)=\sum_{i=1}^N p_i U_i\rho U_i^*$ with $\sum_{i=1}^N p_i=1$ and $U_i\in U(H), 1\leq i\leq N$. Then $\Phi$ preserves $\left\{|i\rangle, 0\leq i\leq 2\right\}$ mutually distinguishable if and only if $U_j^*U_i$ are unitary diagonal matrices for any $1\leq i, j\leq N$. 
\end{corollary}
	
\noindent A generalization of the special type of quantum channels in Corollary \ref{qutrit, basis} is the Schur channel. \cite{watrous2018, levick2017, zhang2022,jochym2013}. Recall that the quantum channel $\Phi_S$ is a Schur channel or sometimes a Schur map over $d-$dimensional quantum system $H$ if there exists an operator $S\in B(H)$ such that, for any $\rho \in D(H)$, we have $\Phi_S(\rho)=S\odot \rho$, where $\odot$ denotes the entry-wise product between $S$ and $\rho$, $(S\odot \rho)(i,j)=S(i,j)\rho (i,j)$ for all $(i,j)-$entries, $1\leq i,j \leq d$. The following theorem shows that the Schur channel $\Phi_S$  would be powerful enough to preserves the basis states $E_d=\{|i\rangle, 1\leq i\leq d\}$ mutually distinguishable.
\begin{theorem}\label{9}\cite{watrous2018}
	Let $H$ be a $d$-dimensional quantum system and $\Phi$ be a quantum channel. The following are equivalent:\\
	\\
	(i) $\Phi$ is a Schur channel for operator $S\in B (H)$. \\
	(ii) There exists a Kraus representation of $\Phi$ having the form $\Phi(\rho)=\sum_{i=1}^NS_i\rho S_i^*$
	with $S_i\in B(H)$ diagonal for each $1\leq i\leq N$ and $\rho\in D(H)$.
	 \\
	(iii) For any Kraus representation of the quantum channel $\Phi$ with $\Phi(\rho)=\sum_{i=1}^NA_i\rho A_i^*$, the Kraus operators $A_i$ are all diagonal for each $1\leq i\leq N$.\\  
\end{theorem}
\noindent For subset $S$ of the standard basis $E_d$. We obtain the following results.
\begin{theorem}\label{10}
	Let $\Phi$ be a mixed unitary quantum channel over $d-$dimensional quantum system $H$. Then $S\subset E_d$ is distinguishable under $\Phi$ if and only if there exist a Kraus representation $\Phi(\rho)=\sum_{i=1}^rA_i\rho A_i^*$ such that $A_j^*A_i(k,l)$ are zeros for any $1\leq i\neq j\leq r$ and $1\leq k\neq l\leq s$. 
\end{theorem}

\noindent A standard method for discussing quantum channels in higher-dimensional spaces is through tensor products of channels defined in qubit systems. A quantum channel $\Phi$ over an N-qubits system $H=H_2^{\otimes N}$ is  \textit{uncorrelated} or quantum channels without memory if 
\begin{equation}\label{2.28}
	\Phi=\Phi_1\otimes \dots \otimes \Phi_N,
\end{equation}  
where $\Phi_i$ is a single qubit quantum channel over the $i-th$ subsystem $H_2$. A quantum channel $\Phi$ over $N$-qubits system is a \textit{correlated channel} or a memory channel if such decomposition of $\Phi$ in Equation \ref{2.28} does not exist. \\
\\The preservation of fidelity between distinguishable states in an uncorrelated channel is clarified in the following discussion. For simplicity, we choose the pure states $|0\dots 001\rangle=|0\rangle\dots|0\rangle|0\rangle |1\rangle$ and $|0\dots 000\rangle=|0\rangle\dots|0\rangle|0\rangle |0\rangle$, where $|0\rangle, |1\rangle$ are from the standard basis of the qubit system $H_2$.
\begin{proposition}\label{tensor channel}
	Let $\Phi=\Phi_1\otimes \cdots\otimes \Phi_N$ be an uncorrelated N-qubits quantum channel over $H=H_2^{\otimes n}$. Then $\Phi$ preserves the states $|\varphi_1\rangle=|0\dots 000\rangle$, $|\varphi_2 \rangle=|0\dots 001\rangle$ distinguishable if and only if $A_j^*A_i$, $1\leq i, j\leq r$, are diagonal for any set of Kraus operators $A_1,\dots, A_r$ of $\Phi_N$.
\end{proposition} 
\noindent The proof is provided in the Appendix. Simultaneously, we note that correlated quantum channels exist extensively with different kinds of correlations over subsystems and the preservation of distinguishable states is still possible. The following is a good example of a two-qubit system that introduces a  correlation between the phase damping processes using a $C_{NOT}$ operation.
  
\begin{example}\label{correlated}
	 The phase damping channel in the qubit system is given by $\Phi(\rho)=(1-\lambda)\rho+\lambda Z\rho Z$ with the Kraus operator $E_0=\sqrt{1-\lambda}I, E_1=\sqrt{\lambda}Z$. For any $\rho\in H\otimes H$, the \textit{Controlled phase damping channel} is defined as 
	\begin{equation}
		\Phi_{CPD}(\rho)=\sum_{i,j=0}^1 A_{ij}\rho A_{ij}^*,
	\end{equation}
	where $A_{ij}=E_i\otimes E_j C_{NOT}$, $0\leq i,j \leq 1$.\\
	\\
	It is straightforward to verify that the controlled phase damping channel preserves the standard basis of this two-qubit system, because the bit flip operation introduced by $C_{NOT}$ preserves the orthogonality between these states, and the phase damping process for each subsystem does not introduce any decoherence. Meanwhile, we can see that the controlled phase damping channel $\Phi_{CPD}$ is not just a tensor of the qubit channels (see Appendix).
\end{example}

\section{Fidelity Preservation for Non-distinguishable Quantum States}

\noindent We now consider the non-distinguishable case, where
$
0<F(\rho_1,\rho_2)<1.
$
This regime is subtler than the distinguishable case, because the problem is no longer preservation of orthogonality alone, but rather characterization of equality in the monotonicity inequality
\[
F(\Phi(\rho_1),\Phi(\rho_2))\geq F(\rho_1,\rho_2).
\]
In other words, we seek conditions under which the channel changes neither the overlap magnitude nor the relevant relative geometric structure for a chosen pair of pure states.
\begin{lemma}\label{13}
	Let $|\varphi_1\rangle, |\varphi_2\rangle$ be non-distinguishable states in $d-$dimensional quantum system $H$ and $\Phi$ be a mixed unitary  quantum channel. Then  $\Phi$ preserves fidelity between $|\varphi_1\rangle$ and $|\varphi_2\rangle$ if and only if 
	\begin{equation}
		||Tr_H(|\Psi_2\rangle\langle\Psi_1|)||_1=|\langle \varphi_1|\varphi_2\rangle|,
	\end{equation}
	where $|\Psi_1\rangle$ and $|\Psi_2\rangle$ are purifications of $|\varphi_1\rangle, |\varphi_2\rangle$ under $\Phi$.
\end{lemma}\

\noindent For a given mixed unitary quantum channel $\Phi(\rho)=\sum_{i=1}^Np_iU_i\rho U_i^*$, we say $\Phi_S$ is a local operation of $\Phi$  with length $s$ if for any quantum state $\rho$, $\Phi_S(\rho)=\sum_{k=i_1}^{i_s}s_kU_k\rho_i U_k^*,$ where $U_k\in S\subseteq K$ and $s_k=\frac{p_k}{\sum_{k=i_1}^{i_s}p_k}$. 
It is clear that any local operation $\Phi_S$ of $\Phi$ under $K$ is still a mixed unitary quantum channel over $H$. To preserve the fidelity between pairs of pure states $|\varphi_1\rangle$ and $|\varphi_2\rangle$, the following lemmas are essential which shows that any local operations $\Phi_S$ of channel $\Phi$ should preserve the fidelity as well.
\begin{lemma}\label{14}\cite{watrous2018}
	Let $H$ be a $d-$dimensional quantum system. Take the density operators $\rho_1, \rho_2, \sigma_1, \sigma_2$ from $H$ and $0\leq \lambda\leq 1$. It holds that
	\begin{equation}
		F(\lambda\rho_1+(1-\lambda)\rho_2, \lambda\sigma_1+(1-\lambda)\sigma_2)\geq \lambda F(\rho_1, \sigma_1)+ (1-\lambda) F(\rho_2, \sigma_2).
	\end{equation} 
\end{lemma}

\begin{lemma}\label{15}
	Let $|\varphi_1\rangle$, $|\varphi_2\rangle$ be non-distinguishable states from a $d-$dimensional quantum system $H$. Suppose that $\Phi$ is a mixed unitary quantum channel with Kraus decomposition $\Phi(\rho)=\sum_{k=1}^Np_kU_k\rho_i U_k^*$ which preserves the fidelity between $|\varphi_1\rangle$ and $|\varphi_2\rangle$. Then, for any local operation $\Phi_S$ of $\Phi$ under $K=\{U_k, 1\leq k\leq N\}$, $\Phi_S$ must preserve the fidelity between these states. (Proof in Appendix)
\end{lemma}
\noindent Further, for local operations $\Phi_{S}$ of $\Phi$ and a shorter local operation $\Phi_{T}$ of $\Phi_{S}$, preserving fidelity between quantum states $|\varphi_1\rangle$ and $|\varphi_2\rangle$ by $\Phi_{S}$ indicates that $\Phi_{T}$ also preserves fidelity between the given states. Then, without loss of generality, for any mixed unitary quantum channel $\Phi(\rho)=\sum_{k=1}^Np_kU_k\rho_i U_k^*$, we always have 
\begin{equation}
	\Phi(\rho)=\sum_{i=1}^l p_{S_i}\Phi_{S_i}(\rho),
\end{equation} 
where $\Phi_{S_i}, 1\leq i\leq l$ are the length-2 local operations of $\Phi$ defined by different Kraus operators, and $l=\left\lfloor \frac{N}{2} \right\rfloor  +1$.\\

\noindent By using Lemma \ref{16} and Lemma \ref{17} below, we can observe from Theorem \ref{18} that the correlations between the Kraus operators for each length-2 local operations $\Phi_{S_i}$ are essential for the ability of the original quantum channel $\Phi$ to preserve fidelity. 
\begin{lemma}\label{16}
	Let M be a $2\times 2$ complex matrix given by $M=z_0I+z_1\sigma_x+z_2\sigma_y+z_3\sigma_z$ under Pauli basis $\{I, \sigma_x, \sigma_y, \sigma_z\}$ and $z_i\in \mathbb{C}$ be complex numbers that parameterize the matrix $M$. Then the singular values $\sigma_{1}$, $\sigma_{2}$ of $M$ are given by
	\begin{equation}
		\sigma_{1,2}=\sqrt{|z_0|^2+|z_1|^2+|z_2|^2+|z_3|^2\pm\sqrt{(|z_0|^2+|z_1|^2+|z_2|^2+|z_3|^2)^2-|z_0^2-z_1^2-z_2^2-z_3^2|^2}}.
	\end{equation}
\end{lemma}\
\\
\\
\noindent For quantum states $|\varphi_1\rangle$, $|\varphi_2\rangle\in H$ and $U\in U(H)$, we define the \textit{relativity} $R_U(|\varphi_1\rangle, |\varphi_2\rangle)$ between $|\varphi_1\rangle$, $|\varphi_2\rangle$ under $U$ as 
\begin{equation}
	R_U(|\varphi_1\rangle, |\varphi_2\rangle)=\frac{\langle \varphi_1 U\varphi_2\rangle}{\langle \varphi_1, \varphi_2\rangle}.
\end{equation}
Further, we say that $|\varphi_1\rangle$, $|\varphi_2\rangle$ are \textit{symmetric} under $U$ (or $U-$symmetric) if $R_U(|\varphi_1\rangle, |\varphi_2\rangle)=R_U(|\varphi_2\rangle, |\varphi_1\rangle)$.  
\begin{lemma}\label{17}
	Let $c, d$ be complex numbers. The following are equivalent:\\
	(i) $c, d$ are unit complex numbers and $c=d^*$.\\
	(ii) $|c|^2+|d|^2\leq 2$ and $2|c-d^*|=||c|^2-|d|^2|$.
\end{lemma}
\noindent We then present our main results:
\begin{theorem}\label{18}
	Let $\Phi$ be a mixed unitary quantum channel in $d-$dimensional quantum system $H$ with the kraus representation $\Phi(\rho)=tU_1\rho U_1^*+(1-t)U_2\rho U_2^*$ for $U_1, U_2\in U(H)$, $0\leq t\leq 1$ and $\rho\in D(H)$. For non-distinguishable pure states $|\varphi_1\rangle$, $|\varphi_2\rangle$, $\Phi$ preserves fidelity between $|\varphi_1\rangle$ and $|\varphi_2\rangle$ if and only if $|\varphi_1\rangle$, $|\varphi_2\rangle$ are symmetric under $U=U_1^*U_2$ and $|R_U(|\varphi_1\rangle, |\varphi_2\rangle)|=1$. 
\end{theorem}

\noindent For details on lemma \ref{17} and theorem \ref{18}, see the Appendix. Based on the restriction from Theorem \ref{18}, we can now construct quantum channels with Choi's rank greater than 2 that preserve the fidelity between non-distinguishable states. Please see the following example. 
\begin{example}\label{3.1}
	Let $H$ be a qubit system and $|\varphi_1\rangle=\begin{pmatrix}
		1\\ 0
	\end{pmatrix}$, $|\varphi_2\rangle=\begin{pmatrix}
		a\\ b
	\end{pmatrix}$ be non-distinguishable pure states with $|a|^2+|b|^2=1$. And $\Phi$ is the the mixed unitary quantum channel given by
	\begin{equation}
		\Phi(\rho)=\frac{1}{3}U_1\rho U_1^*+\frac{1}{3}U_2\rho U_2^*+\frac{1}{3}U_3\rho U_3^*,
	\end{equation}
where $U_1=\begin{pmatrix} 1&0\\ 0&1\end{pmatrix}$, $U_2=\begin{pmatrix} 1&0\\ 0&i\end{pmatrix}$, $U_3=\begin{pmatrix} 1&0\\ 0&-1\end{pmatrix}$.\\
It is not difficult to verify that $\Phi$ preserves the non-distinguishable states $|\varphi_1\rangle,|\varphi_2\rangle$.\\
\\ This example shows that fidelity preservation for non-distinguishable states can occur for mixed unitary channels with Choi rank greater than $2$, but only for specially aligned pairs. It should therefore be viewed as a concrete construction illustrating the restriction from Theorem 18, rather than as evidence that such preservation is generic.
\end{example}

\section{General Phase Damping Channels}

\noindent In Section 3, we observed that preservation of distinguishability is closely tied to diagonal structure in the relative products $U_j^*U_i$. In Section 4, we saw that channels of this type can also preserve the fidelity of certain non-distinguishable pure-state pairs. We now study \emph{general phase damping} (GPD) channels as a representative case in which the equality problem for fidelity can be examined transparently.\\
\\
More specifically, if a mixed unitary channel has the property that the relative products $U_1^*U_i$ are diagonal, then, up to an overall unitary conjugation, the relevant behavior is governed by a channel with diagonal unitary Kraus operators. This motivates the following definition of a general phase damping channel.\\
\\
We emphasize that the goal of this section is not to develop a general theory of decoherence, but rather to understand when equality can occur in the fidelity monotonicity relation within a natural class of channels whose generic action is to contract off-diagonal terms.
\begin{definition}
	Let $H$ be an $n-$dimensional quantum system. We say $\Phi$ is a general phase damping channel (GPD) over $H$ if there exists a Kraus representation of $\Phi$ such that 
 \begin{equation}
	\Phi(\rho)=p_1\rho+\sum_{i=2}^N p_iD_i\rho D_i^*
 \end{equation}
where $D_i=diag\{1, e^{i\theta_i}\}, 2\leq i\leq N $, for any $\rho\in D(H)$ and some relative phases $\theta_i$ and $p_i$ satisfying $\sum_{i=1}^N p_i=1$.
\end{definition}
\noindent We now consider general phase damping channels (GPD) as a case study that connects the two themes of this study: exact fidelity preservation and decoherence. Recall that fidelity is monotone under quantum channels. Therefore, for generic pairs of states, one expects.
\begin{equation}
	F(\Phi(\rho_1), \Phi(\rho_2))\geq F(\rho_1, \rho_2),
\end{equation}
The central question is to characterize when the equality case holds. In a GPD channel, the diagonal structure of the Kraus operators makes the mechanism transparent: the channel primarily contracts off-diagonal terms, that is, it introduces decoherence. The goal of this section is twofold: (i) to record structural constraints on families of pure states whose pairwise fidelities are exactly preserved, and (ii) to visualize how decoherence generically drives the inequality to be strict, thereby explaining why the preserved families must be small.\\
\\
 Let $p=\prod_{i=0}^Np_i$ denote the mixing parameter among the unitary operations $\{U_0=I, U_1, \dots U_N\}$ of $\Phi$ and $N$ the general dephasing rank. It is evident that $0\le p\le (\frac{1}{N})^N$. Let $E_d=\{|0\rangle, \dots, |d-1\rangle\}$ be the mutually distinguishable set of the $d-$dimensional quantum system. Then the following result is obtained directly from Theorem \ref{2d}.\\
\begin{corollary}\label{19}
	Let $H$ be a $d-$dimensional quantum system and $\Phi$ a quantum channel over $H$. Then $E_d$ is distinguishable under $\Phi$ if and only if $\Phi$ is unitary equivalent to a general phase damping channel.
\end{corollary}

\noindent Example \ref{3.1} demonstrates that such a quantum channel also preserves the fidelity between non-distinguishable pure states. A natural question arises: what is the largest set \( S \) of pure states that is preserved by $ \Phi$? For the case of a qubit system $ H $, we have the following result:
\begin{theorem}\label{20}
	Let $H$ be a qubit system and $\Phi$ be a general phase damping channel over $H$. If a set of pure states $S$ is preserved by $\Phi$, then $ |S| \leq 3 $. Moreover, $ |S| = 3 $ if and only if $ E_2 \subset S $.\end{theorem}

\noindent This result shows that even within this highly structured channel class, exact fidelity preservation cannot hold on a large generic family of pure states. Thus pairwise fidelity preservation remains a rigid phenomenon, even in a setting where the form of decoherence is especially simple (See the proof of Theorem \ref{20} in the Appendix). One of the reasons why we name the quantum channel $\Phi$ in Section 4.1 as a general phase damping channel is because the classical phase damping channel is indeed a special case of $\Phi$ with $n=2, \theta_2=\pi$ and its phase damping parameter $\lambda$ can be defined by probabilities $\{p_1, p_2\}$ (See the end of Appendix)\\

\noindent We now discuss why preserving fidelity between an arbitrary pair of quantum states is generally difficult. We begin with a phase damping channel. $\Phi_P$ is referred to as a phase damping channel in a qubit system with dephasing parameter $0\leq \lambda \leq 1$ if for any $\rho\in D(H)$, $\Phi_{P}(\rho)=K_1\rho K_1^*+K_2\rho K_2^*$, where
\begin{equation}
		K_1=\begin{pmatrix}{}
  1&0 \\
  0&\sqrt{1-\lambda} 
\end{pmatrix}, 
K_2=\begin{pmatrix}{}
  0&0 \\
  0&\sqrt{\lambda} 
\end{pmatrix}.
	\end{equation}
Similarly (See Appendix), $\Phi_P(\rho)=p_1\rho+(1-p_1)D\rho D^*$ with $D=\diag\{1, e^{i\theta}\}$ for relative phase $\theta$ and mixing parameter $p=p_1(1-p_1)$.\\ 
\\
For any density operator 
\begin{equation}
	\rho=\begin{pmatrix}{}
  a_{11}&a_{12} \\
  a_{12}^*&a_{22} 
\end{pmatrix},
\end{equation}
we have
\begin{equation}
	\Phi_P(\rho)=\begin{pmatrix}{}
  a_{11}&p_1a_{12}+(1-p_1)a_{12}e^{-i\theta} \\
  p_1a_{12}^*+(1-p_1)a_{12}^*e^{i\theta}&a_{22} 
\end{pmatrix}.
\end{equation}
This introduces decoherence of the input states because $|p_1a_{12}+(1-p_1)a_{12}e^{-i\theta}|\leq |a_{12}|$. Because fidelity is monotone under CPTP maps, decoherence cannot decrease fidelity. For generic pairs, this typically makes the inequality strict unless the pair lies in an equality case.\\

\noindent The decoherence phenomenon introduced by the phase-damping channel with respect to the mixing parameter $p$ is shown in Figure 5.1, by taking the maximally coherent state $\rho_m$ as the input, where
\begin{equation}
	\rho_m=\frac{1}{2}\begin{pmatrix}{}
  1&1 \\
  1&1 
\end{pmatrix}.
\end{equation}
\begin{figure}[ht]
  \centering
  \includegraphics[width=0.5\textwidth]{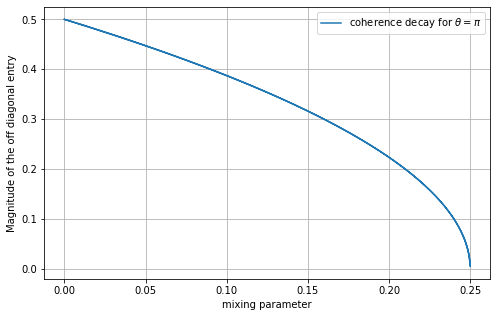}
  \caption{Coherence decay for Phase damping channel}
  \label{fig:my_label}
\end{figure}
\\
\\
For a rank 2 general phase damping channel $\Phi(\rho)=p_1\rho+(1-p_1)D\rho D^*$ with $D=\diag\{1, e^{i\theta}\}$, relative phase $\theta$ and mixing parameter $p=p_1(1-p_1)$, the decoherence process can be shown in a very similar manner. See Figure 5.2.
\\
\begin{figure}[t]
  \centering
  \includegraphics[width=0.5\textwidth]{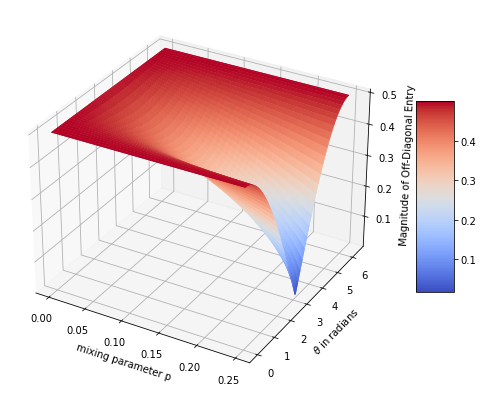}
  \caption{Coherence decay under a rank 2 general phase damping error with relative phase $\theta$ and mix parameter p}
  \label{fig:my_label}
\end{figure}

\noindent We can also visualize the decoherence for the rank 3 general phase damping channel over a qubit system defined by $\Phi(\rho)=(1-p_2-p_3)\rho+p_2D_2\rho D_2^*+p_3D_3\rho D_3^*$ with diagonal matrices $D_2=\diag\{1, e^{i\theta_2}\}$,  $D_3=\diag\{1, e^{i\theta_3}\}$ for $\theta_2, \theta_3\in \mathcal{R}$, probabilities $p=\{1-p_2-p_3, p_2, p_3\}$ and the mixing parameter $p=(1-p_2-p_3)p_2p_3$. See Figures 5.3, 5.4 below.\\
\begin{figure}[ht]
  \centering
  \begin{subfigure}{0.5\textwidth}
    \centering
    \includegraphics[width=\linewidth]{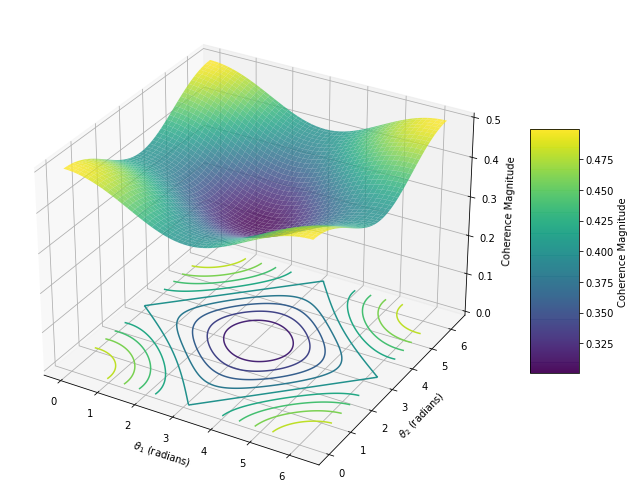}
    \caption{Decoherence in rank-3 GPD with \\probabilities $p=\{0.8, 0.1, 0.1\}$}
    \label{fig:sub1a}
  \end{subfigure}\hfill
  \begin{subfigure}{0.5\textwidth}
    \centering
    \includegraphics[width=\linewidth]{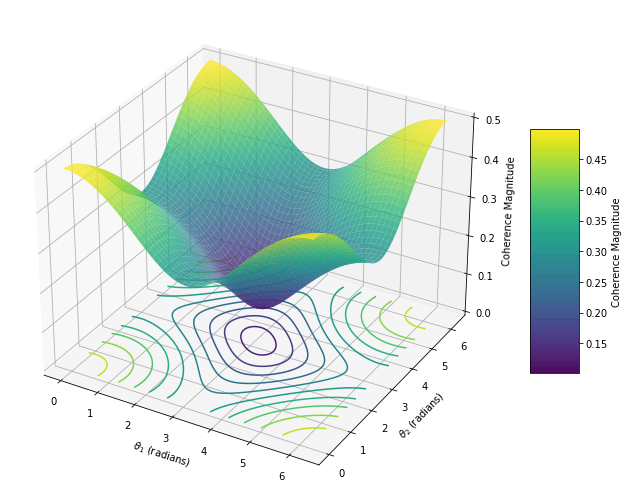}
    \caption{Decoherence in rank-3 GPD with \\probabilities $p=\{0.6, 0.2, 0.2\}$}
    \label{fig:sub2a}
  \end{subfigure}
  \caption{}
  \label{fig:test}
\end{figure}
\\
\begin{figure}[ht]
  \centering
  \begin{subfigure}{0.5\textwidth}
    \centering
    \includegraphics[width=\linewidth]{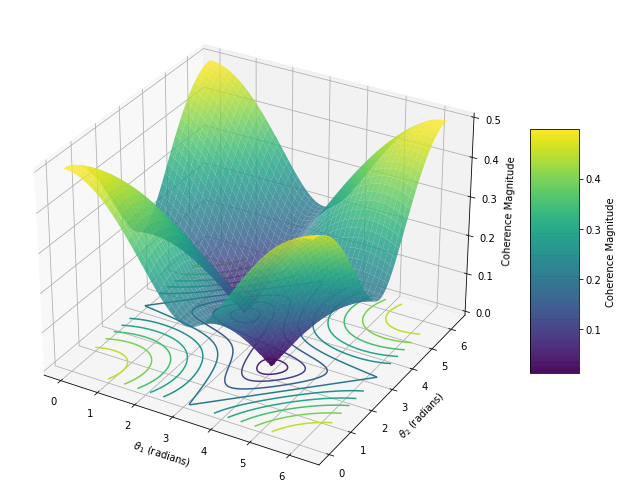}
    \caption{Decoherence in rank-3 GPD with \\probabilities $p=\{0.4, 0.3, 0.3\}$}
    \label{fig:sub1b}
  \end{subfigure}\hfill
  \begin{subfigure}{0.5\textwidth}
    \centering
    \includegraphics[width=\linewidth]{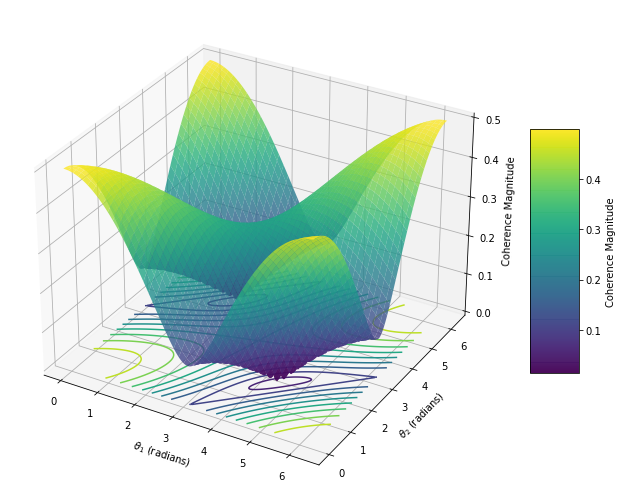}
    \caption{Decoherence in rank-3 GPD with \\probabilities $p=\{0.2, 0.4, 0.4\}$}
    \label{fig:sub2b}
  \end{subfigure}
  \caption{}
  \label{fig:test}
\end{figure}
The GPD case therefore provides a concrete setting in which the difference between generic decoherence and exact fidelity preservation becomes especially visible: while off-diagonal damping is typical, equality in fidelity monotonicity survives only under strong structural restrictions.
\clearpage
\section{Summary and Outlook}

\noindent In this paper, we studied exact preservation of the Uhlmann fidelity for selected pairs of pure states under mixed unitary quantum channels. Our focus was on the equality
\[
F(\rho_1,\rho_2)=F(\Phi(\rho_1),\Phi(\rho_2)),
\]
viewed as a state-dependent preservation problem under the channel itself. This notion is weaker than full state protection through quantum error correction, decoherence-free subspaces, or noiseless subsystems, but it is still sufficiently rigid to impose strong structural constraints.\\
\\
For distinguishable pure-state pairs, we characterized when orthogonality, equivalently zero fidelity, is preserved under mixed unitary channels. In the qubit case, this yielded an exact criterion in terms of the diagonal structure of the relative products $U_j^*U_i$, and we extended this perspective to several higher-dimensional settings. For non-distinguishable pure-state pairs, we studied equality in the monotonicity of fidelity and obtained a preservation criterion for two-unitary mixtures in terms of a symmetry relation involving the relative unitary $U_1^*U_2$.\\
\\
We further analyzed general phase damping channels as a concrete case study in which exact fidelity preservation can be compared with generic decoherence. In this class, off-diagonal contraction typically forces strict monotonicity, so equality can occur only for restricted families of states. In particular, this shows that pairwise fidelity preservation may survive in structured noisy dynamics even when full noiseless behavior is absent.\\
\\
There are several directions for future work. One is to investigate whether similar preservation criteria can be obtained for broader classes of quantum channels beyond the mixed unitary setting. Another is to study preservation of other quantum properties within the same framework of quantum property preservation. It would also be interesting to understand how the structure of preserved state families interacts with channel invariants and with other notions of partial information protection.\newpage

\section{Appendix}
\begin{proof}[Proof of Theorem \ref{2d}:]

	Let $\rho_1, \rho_2\in D(H)$ be the associated density operators of $|\varphi_1\rangle, |\varphi_2\rangle$, respectively.\\
	 We first show (i) $\Rightarrow$ (ii). 
	Suppose that $\Phi(\rho_1)$ and $\Phi(\rho_2)$ are still distinguishable under quantum channel $\Phi$. That indicates
	\begin{equation}
		F\left(\Phi(\rho_1), \Phi(\rho_2)\right)=0.
	\end{equation}
Let $\Psi$ be the complementary channel of $\Phi$ from $D(H)$ to $D(H)$. For each $\rho\in D(H)$, we have
	\begin{equation}
		\Psi(\rho)=\sum_{i, j=1}^N\sqrt{p_ip_j}\langle U_j^*U_i, \rho \rangle E_{ij}.
	\end{equation}

\noindent Without loss of generality, we take the initial distinguishable states $|\varphi_1\rangle=|0\rangle, |\varphi_2\rangle=|1\rangle$ from the Qubit system $H$ with associated density operators $\rho_1=|0\rangle\langle0|$, $\rho_2=|1\rangle\langle 1|$ (See remark 1). By Lemma \ref{nuclear norm}, we have
\begin{equation}
	||\Psi(|1\rangle \langle0|)||_1=F(\Phi(\rho_1), \Phi(\rho_2))=0, 
\end{equation}
which indicates $\Psi(|1\rangle\langle0|)=0$ and then 
\begin{equation}\label{2.10}
	\langle U_j^*U_i, |1\rangle\langle0|\rangle=\langle U_j^*U_i, E_{12}\rangle=0
\end{equation}
 for any $1\leq i, j\leq N$.\\
Clearly, $U_j^*U_i$ are still unitary matrices in $B(H)$. Thus, for each choice of $(i, j)$, we can write 
\begin{equation}
	U_j^*U_i=
	\begin{pmatrix}{}
  a&b \\
  -e^{i\theta}\bar{b}& e^{i\theta}\bar{a}
\end{pmatrix}
\end{equation} 
for some $\theta\in\mathcal{R}$ and $a, b$ complex numbers satisfying $|a|^2+|b|^2=1$. Equation \ref{2.10} implies that $b=0$. Thus $U_j^*U_i$ is generally diagonal for some $\theta_a$ and $\theta_b$. i.e.
\begin{equation}
	U_j^*U_i=
\begin{pmatrix}{}
  e^{i\theta_a}&0 \\
  0& e^{i\theta_b}
\end{pmatrix}.
\end{equation}
\\
\noindent Next, we show (ii)$\Rightarrow$(i). Suppose that Kraus operators $\{U_i\}_{1\leq i\leq N}$ of $\Phi$ satisfy
\begin{equation}
	U_j^*U_i= \begin{pmatrix}{}
  e^{i\theta_1^{ij}}&0 \\
  0& e^{i\theta_2^{ij}}.
\end{pmatrix}
\end{equation} 
for some $\{\theta_1^{ij}, \theta_2^{ij}\}$. Then it is easy to verify that
\begin{equation}
	\begin{aligned}
		\langle \Phi(\rho_1), \Phi(\rho_2)\rangle=0=\langle \rho_1, \rho_2\rangle,
	\end{aligned}
\end{equation}
  which shows that the quantum channel $\Phi$ preserves the distinguishable states.\\
\noindent (iii) $\Rightarrow$ (ii) is clear if we take $A_i=\sqrt{p_i}U_i$.\\
(ii) $\Rightarrow$ (iii) is guaranteed by Theorem \ref{Choi-thm2}. Let $\Phi(\rho)=\sum_{i=1}^MA_i\rho A_i^*$ be any other Kraus representation of $\Phi$. Theorem \ref{Choi-thm2} shows that  there exists an isometry $U_{M\times N}$ between $(\sqrt{p_1}N_1,\dots \sqrt{p_N}N_N)^T$ and $(A_1,\dots A_M)^T$ such that for any $1\leq i\leq M$,
\begin{equation}
	A_i=\sum_{s=1}^N u_{is}\sqrt{p_s}U_s.
\end{equation}
Thus $A_j^*A_i$ are still diagonal matrices for any $1\leq i, j\leq M$ since
\begin{equation}
	A_j^*A_i=\sum_{s, t=1}^N u_{is}u_{jt}\sqrt{p_s}\sqrt{p_t}U_t^*U_s,
\end{equation}
which completes the proof. 

\end{proof}

\begin{proof}[Proof of lemma \ref{qutrit}]
	Let $U=(u_{ij})_{1\leq i,j\leq 3}$ be a unitary matrix in a qutrit system. It is clear that $U$ is a two-level unitary matrix if it possesses more than one pair of off-diagonal zeros. Thus, it is sufficient to discuss $U$ in the following form, 
\begin{equation}
	U=\begin{pmatrix}{}
  u_{11}& 0 &u_{13} \\
  0& u_{22}&u_{23} \\
  u_{31}&u_{32} &u_{33}
\end{pmatrix},
\end{equation} 
that has only one pair of off-diagonal zeros.\\

\noindent Because $U^*U=I$, we have
\begin{equation}\label{5.2}
	\begin{aligned}
		I&=\begin{pmatrix}{}
  u_{11}^*& 0 &u_{31}^* \\
  0& u_{22}^*&u_{32}^* \\
  u_{13}^*&u_{23}^* &u_{33}^*
\end{pmatrix}\begin{pmatrix}{}
  u_{11}& 0 &u_{13} \\
  0& u_{22}&u_{23} \\
  u_{31}&u_{32} &u_{33}
\end{pmatrix}\\
\\
		&=\begin{pmatrix}{}
  u_{11}^2+u_{13}^2& u_{31}^*u_{32} &u_{11}^*u_{13}+u_{31}^*u_{33} \\
  u_{32}^*u_{31}& u_{22}^2+u_{32}^2&u_{11}^*u_{13}+u_{31}^*u_{13} \\
  u_{33}^*u_{31}+u_{13}^*u_{11}&u_{23}^*u_{22}+u_{23}^*u_{32} &u_{13}^2+u_{22}^2+u_{33}^2
\end{pmatrix}.
	\end{aligned}
\end{equation}
\\
Then we have $u_{31}^*u_{32}=0$, which indicates that at least one of $u_{31}$ and $u_{32}$ equals zero. Without loss of generality, we take $u_{31}=0$ and then we have $u_{11}^*u_{13}=0$ using Equation \ref{5.2}. Here, $u_{11}=0$ will result in a contradiction with the assumption that $U$ is a unitary matrix. Thus $u_{13}=0$ and enforces $U$ to be block diagonal. Furthermore, if we take $u_{32}=0$ instead of $u_{31}$, it will introduce a two-level unitary matrix $U$ that acts on another two vector components. This completes the proof.
\end{proof}

\begin{proof}[Proof of theorem \ref{tensor channel}]
	From Lemma \ref{Landau}, we can see that any unital quantum channel $\Phi$ over a single Qubit system is mixed unitary. Thus, for each $\Phi_i$, we take 
	$\Phi_i(\rho)=\sum_{j=1}^{N_i}p_{ij}U_{ij}\rho U_{ij}^*$ with $\sum_{j=1}^{N_i}p_{ij}=1$.\\
	
	\noindent We follow the approach in Theorem \ref{2d} to prove Theorem 12 by examining the properties of the complementary quantum channel $\Psi$ of $\Phi$, where
	\begin{equation}
		\Phi(\rho)=\sum p_{j_1j_2\cdots j_N}U_{1j_1}\otimes U_{2j_2}\otimes \cdots \otimes U_{Nj_N}\rho U_{1j_1}^*\otimes U_{2j_2}^*\otimes \cdots \otimes U_{Nj_N}^*	\end{equation}
where $p_{j_1j_2\cdots j_N}=p_{1j_1}p_{2j_2}\cdots p_{Nj_N}, 1\leq j_i\leq N_i$ for any $\rho\in D(H_N)$.\\
\\
Let $\rho_1=|0\rangle\langle0|$ and $\rho_2=|0\rangle\langle1|$. 
Then $\Phi$ preserves the distinguishable states $|\varphi_1\rangle$, $|\varphi_2\rangle$ if and only if $||\Psi(|\varphi_1\rangle\langle\varphi_2|)||_1=0$ or the followings always equals zero for different choices of the indices:
\begin{equation}
	\begin{aligned}
		\langle (U_{1j_1}^*\otimes &\cdots \otimes U_{Nj_N}^*)(U_{1i_1}\otimes \cdots \otimes U_{Ni_N}), |\varphi_1\rangle\langle\varphi_2|\rangle\\
		&=\langle (U_{1j_1}^*U_{1i_1})\otimes \cdots \otimes(U_{Nj_N}^*U_{Ni_N}), \rho_1\otimes \rho_1\otimes\cdots\otimes \rho_2\rangle\\
		&=\langle U_{1j_1}^*U_{1i_1}, \rho_1 \rangle \cdots \langle U_{(N-1)j_{N-1}}^*U_{(N-1)i_{N-1}}, \rho_1 \rangle \langle U_{Nj_N}^*U_{Ni_N}, \rho_2 \rangle.
	\end{aligned}
\end{equation}
We can observe that $\langle U_{Nj_N}^*U_{Ni_N}, \rho_2 \rangle=0$ should be true for any choice of the unitary operators. This indicates that $U_{Nj}^*U_{Ni}$ is always diagonal and then $A_j^*A_i$ are also diagonal for any $1\leq i, j\leq r_N$, which completes the proof.
\end{proof}

\begin{proof}[Proof of example \ref{correlated}]
	Suppose that the controlled phase damping channel $\Phi_(\rho)=\sum_{i,j=0}^1 A_{ij}\rho A_{ij}^*$ in a two-qubit system is generated by the tensor product of the two qubit quantum channels. Take quantum channels $\Phi_1(\rho)=A_1\rho A_1^*+A_2\rho A_2^*$ and  $\Phi_2(\rho)=B_1\rho B_1^*+B_2\rho B_2^*$ such that 
\begin{equation}
	A_{ij}=A_i\otimes B_j
\end{equation} 
for all $1\leq i,j\leq 2$.\\
\\
Now, consider the case when $i,j=1$, we have 
\begin{equation}
	A_{11}=(1-\lambda)I\otimes I \cdot C_{NOT}=A_1\otimes B_1
\end{equation}
which implies
\begin{equation}\label{5.5}
	(1-\lambda)\begin{pmatrix}{}
  1&0&0 &0\\
  0& 1& 0&0\\
  0&0& 0&1\\
  0& 0& 1&0
\end{pmatrix}
=\begin{pmatrix}{}
  a_1& b_1\\
  c_1& d_1
\end{pmatrix}
\otimes
\begin{pmatrix}{}
  a_2& b_2\\
  c_2& d_2
\end{pmatrix}
\end{equation}
\\
for $A_1=\begin{pmatrix}{}
  a_1& b_1\\
  c_1& d_1
\end{pmatrix}$ and 
$B_1=\begin{pmatrix}{}
  a_2& b_2\\
  c_2& d_2
\end{pmatrix}$.\\
\\
By considering the left corner of the diagonal block in Equation \ref{5.5}, we have $b_2=c_2=0$ and $a_1, a_2, d_2$ are nonzero. However, an analogous discussion in the right corner of the diagonal block shows that $b_2, c_2$ are nonzero, leadings to a contradiction. This implies that $\Phi_{C_{CPD}}$ is not a tensor product of qubit channels. \\
\end{proof}

\begin{proof}[Proof of lemma \ref{15}]
Let $\rho_1, \rho_2$ be the density operators associated with $|\varphi_1\rangle$ and $|\varphi_2\rangle$ and  $\Phi_S(\rho)=\sum_{k=i_1}^{i_s}s_kU_k\rho U_k^*$ be any length-s local operation of $\Phi$ with $s_k=\frac{p_k}{\sum_{k=i_1}^{i_s}p_k}$. Let $\Phi_{S^C}(\rho)=\sum_{k=j_1}^{j_{r-s}}s_kU_k\rho_i U_k^*$ be the complement channel of $\Phi_S$ defined by the other Kraus operators of $\Phi$. Then we have 
\begin{equation}
	\Phi(\rho)=(\sum_{k=i_1}^{i_s}p_k)\Phi_S(\rho)+(\sum_{k=j_1}^{j_{r-s}}p_k)\Phi_{S^C}(\rho)
\end{equation}
with $\sum_{k=i_1}^{i_s}p_k+\sum_{k=j_1}^{j_{r-s}}p_k=1$.\\
\\
Take $p_S=\sum_{k=i_1}^{i_s}p_k$ and $p_{S^c}=\sum_{k=j_1}^{j_{r-s}}p_k$. Then
\begin{equation}\label{3.6}
	\begin{aligned}
		1&\geq F(\Phi(\rho_1), \Phi(\rho_2))\\
		& =F(p_S\Phi_S(\rho_1)+ p_{S^c}\Phi_{S^c}(\rho_1), p_S\Phi_S(\rho_2)+ p_{S^c}\Phi_{S^c}(\rho_2))\\
		&\geq p_S F(\Phi_S(\rho_1), \Phi_S(\rho_2))+ p_{S^c} F(\Phi_{S^c}(\rho_1), \Phi_{S^c}(\rho_2))\\
		&\geq p_S F(\rho_1, \rho_2)+ p_{S^c} F(\rho_1, \rho_2)\\
		&= F(\rho_1, \rho_2).
	\end{aligned} 
\end{equation}
Then, if the quantum channel $\Phi$ preserves fidelity between $|\varphi_1\rangle$ and $|\varphi_2\rangle$, i.e. 
\begin{equation}
	F(\rho_1, \rho_2)=F(\Phi(\rho_1), \Phi(\rho_2)),
\end{equation}
 all inequalities in Equation \ref{3.6} must achieve equality. This indicates that each local operation $\Phi_S$ also preserves fidelity between $|\varphi_1\rangle$ and $|\varphi_2\rangle$. 
\end{proof}

\begin{proof}[Proof of lemma \ref{17}]
	(i)$\Rightarrow$ (ii) is obvious since $|c|^2+|d|^2=2$ and $|c-d^*|=0=||c|^2-|d|^2|$.\\
	(ii)$\Rightarrow$ (i): For complex numbers $c, d$ that satisfy (ii) , we have 
	\begin{equation}\label{3.11}
		(|c|+|d^*|)^2=|c|^2+|d|^2+2|c||d|\leq 2(|c|^2+|d|^2),
	\end{equation}
which implies 
\begin{equation}\label{3.12}
	|c|+|d^*|\leq \sqrt{2}(|c|^2+|d^*|^2)^{\frac{1}{2}}\leq \sqrt{2}\sqrt{2}=2.
\end{equation}
Thus, 
\begin{equation}\label{3.13}
	\begin{aligned}
		2|c-d^*|&=||c|^2-|d|^2|\\
		&=(|c|+|d|^*)||c|-|d^*||\\
		&\leq 2 ||c|-|d^*||,
	\end{aligned}
\end{equation}
which indicates that $|c-d^*|\leq ||c|-|d^*||$.\\
\\
Simultaneously, for complex numbers $c, d$, it is always true that 
\begin{equation}\label{3.14}
	|c-d^*|\geq ||c|-|d^*||.
\end{equation}
Then we have
	\begin{equation}\label{3.15}
		|c-d^*|= ||c|-|d^*||.
	\end{equation}  
	 By Equations \ref{3.12} and \ref{3.13}, $|c|=|d^*|=1$ and thus $c=d^*$ by Equation \ref{3.15}. 
\end{proof}

\begin{proof}[Proof of theorem \ref{18}]
	Let $\Phi_0(\rho)=t\rho+(1-t)U\rho U^*$, $U=U_1^*U_2$ and suppose $|\langle \varphi_1|\varphi_2\rangle|=\alpha$.  Because fidelity is invariant under unitary operations, we have that $\Phi$ preserves fidelity between $|\varphi_1\rangle$, $|\varphi_2\rangle$ if and only if $\Phi_0$ preserves fidelity between $|\varphi_1\rangle$, $|\varphi_2\rangle$.\\
	
	\noindent By Lemma \ref{13}, it is sufficient to show that $||Tr_H(|\Psi_2\rangle\langle\Psi_1|)||_1=\alpha$ if and only if $|\varphi_1\rangle$ and $|\varphi_2\rangle$ are $U$-symmetric and $|R_U(|\varphi_1\rangle, |\varphi_2\rangle)|=1$, where
	\begin{equation}
	|\Psi_i\rangle=\sum_{k=1}^2\sqrt{p_k}(U_k|\varphi_i\rangle\otimes |k\rangle).
\end{equation}
Let $A=\begin{pmatrix}
	t&b\\
	a&1-t
\end{pmatrix}$, where $a=\sqrt{t(1-t)}|R_U(|\varphi_1\rangle, |\varphi_2\rangle)$ and  $b=\sqrt{t(1-t)}|R_{U^*}(|\varphi_1\rangle, |\varphi_2\rangle)$. Then 
\begin{equation}
	\begin{aligned}
		||Tr_H(|\Psi_2\rangle|\langle\Psi_1|)||_1&=\left|\left|
		\begin{pmatrix}
			t\langle\varphi_1\varphi_2\rangle &\sqrt{t(1-t)}\langle\varphi_1 U^*\varphi_2\rangle\\
			\sqrt{t(1-t)}\langle\varphi_1 U\varphi_2\rangle & (1-t)\langle\varphi_1|\varphi_2\rangle
		\end{pmatrix}\right|\right|_1\\
		&=|\langle\varphi_1\varphi_2\rangle|\left|\left|
		\begin{pmatrix}
			t &b\\
			a & 1-t
		\end{pmatrix}\right|\right|_1\\
		&=\alpha ||A||_1.
	\end{aligned}
\end{equation} 
Thus $\Phi$ preserves fidelity if and only if $||A||_1=1$.
We now consider the case when the nuclear norm of $A$ is 1.\\

\noindent For the matrix $A$, under the Pauli basis $\{I, \sigma_x, \sigma_y, \sigma_z\}$, we have
\begin{equation}
	A=z_0I+z_1\sigma_x+z_2\sigma_y+z_3\sigma_z, 
\end{equation} 
where $z_0=\frac{1}{2}$, $z_1=\frac{a+b}{2}$, $z_2=\frac{a-b}{2i}$, $z_3=t-\frac{1}{2}$.
Then $\Phi$ preserves the fidelity between $|\varphi_1\rangle$ and $|\varphi_2\rangle$ if and only if $||A||_1=1$. Suppose that the singular values of $A$ are $\sigma_{1}$ and $\sigma_{2}$.\\

\noindent By Lemma \ref{16}, we have 
\begin{equation}\label{3.21}
	\begin{cases}\sigma_1^2+\sigma_2^2=2(|z_0|^2+|z_1|^2+|z_2|^2+|z_3|^2)&\\ \sigma_1\sigma_2=|z_0^2-z_1^2-z_2^2-z_3^2|&\end{cases} .
\end{equation}
Then $\Phi$ preserves fidelity if and only if $||A||_1=\sigma_1+\sigma_2=1$, which is equivalent to the following equations by Equation \ref{3.21}.
\begin{equation}
	\begin{aligned}
		&|z_0|^2+|z_1|^2+|z_2|^2+|z_3|^2+|z_0^2-z_1^2-z_2^2-z_3^2|=\frac{1}{2}\\
		&\frac{1}{4}+(t-\frac{1}{2})^2+\left|\frac{a+b}{2}\right|^2+\left|\frac{a-b}{2i}\right|^2+\left|\frac{1}{4}-\left(\frac{a+b}{2}\right)^2-\left(\frac{a-b}{2i}\right)^2-\left(t-\frac{1}{2}\right)^2\right|=\frac{1}{2}\\
		&t^2-t+\frac{1}{2}(|a|^2+|b|^2)+|t-t^2-ab|=0\\
		&|t^2-t+ab|=t-t^2-\frac{1}{2}(|a|^2+|b|^2).
	\end{aligned}
\end{equation}	
The above equation holds if and only if 
\begin{equation}\label{3.22}
	\begin{cases}t-t^2-\frac{1}{2}(|a|^2+|b|^2)\geq 0&\\ 
	|t^2-t+ab|^2=(t-t^2-\frac{1}{2}(|a|^2+|b|^2))^2&\end{cases} .
\end{equation}
Take $a_0=\frac{1}{\sqrt{t(1-t)}}a=R_U(|\varphi_1\rangle, |\varphi_2\rangle)$ and $b_0=\frac{1}{\sqrt{t(1-t)}}b=R_{U^*}(|\varphi_1\rangle, |\varphi_2\rangle)$.
The first equation in Equation \ref{3.22} turns to $|a|^2+|b|^2\leq 2t(1-t)$ and then 
\begin{equation}\label{3.23}
	|a_0|^2+|b_0|^2\leq 2.
\end{equation} 
Further, the second equation from Equation \ref{3.22} is equivalent to the following
\begin{equation}\label{3.24}
	\begin{aligned}
		(t^2-t+ab)(t^2-t+a^*b^*)&=(t^2-t)^2+\frac{1}{4}(|a|^2+|b|^2)^2+(t^2-t)(|a|^2+|b|^2)\\
		|a|^2|b|^2+(t^2-t)(ab+a^*b^*)&=\frac{1}{4}(|a|^2+|b|^2)^2+(t^2-t)(|a|^2+|b|^2)\\
		(t^2-t)(|a|^2+|b^2|-ab-a^*b^*)&=|a|^2|b|^2-\frac{1}{4}(|a|^4+|b|^4+2|a|^2|b|^2)\\
		(t-t^2)|a-b^*|^2&=\frac{1}{4}(|a|^2-|b|^2)^2\\
		2|a_0-b_0^*|&=||a_0|^2-|b_0^*|^2|.
	\end{aligned}
\end{equation}
By Lemma 17, Equation \ref{3.23}, Equation \ref{3.24} are true at the same time if and only if $a_0, b_0$ are unit complex numbers that satisfy $a_0=b_0^*$.
Then $\Phi$ preserves fidelity if and only if  $|R_U(|\varphi_1\rangle, |\varphi_2\rangle)|=1$ and $|\varphi_1\rangle$, $|\varphi_2\rangle$ are symmetric under $U$, which is based on the observation that
\begin{equation}
	b_0^*=R_{U^*}(|\varphi_1\rangle, |\varphi_2\rangle)^*=\left(\frac{\langle\varphi_1 U^* \varphi_2\rangle}{\langle \varphi_1|\varphi_2\rangle}\right)^*=\frac{\langle\varphi_2 U \varphi_1\rangle}{\langle \varphi_2|\varphi_1\rangle}=\frac{\langle\varphi_1 U \varphi_2\rangle}{\langle \varphi_1|\varphi_2\rangle}=a_0
\end{equation}
 and 
\begin{equation}
a_0=\frac{\langle\varphi_1 U \varphi_2\rangle}{\langle \varphi_1|\varphi_2\rangle}.
\end{equation}
\end{proof}

\begin{proof}[Proof of theorem \ref{20}]
	Let $S$ be the largest set of pure states from $H$ that are preserved by a general phase damping channel $\Phi$. For any superposition $|\varphi\rangle =\begin{pmatrix}
		a\\
	    b
	\end{pmatrix}\in S$ with $|a|^2+|b|^2=1$, we can verify that
\begin{equation}\label{4.3}
	F(|0\rangle\langle0|, |\varphi\rangle\langle\varphi|)=|a|=F(\Phi(|0\rangle\langle0|),\Phi( |\varphi\rangle\langle\varphi|)),
\end{equation}
which is also true between the states $|1\rangle$ and $|\varphi\rangle$, thus $E_2\subset S$.\\

\noindent Further, by Lemma \ref{15}, $S$ is preserved by any length-2 local operation $\Phi_0$ of $\Phi$. Without losing generality, take
\begin{equation}
		\Phi_0(\rho)=p\rho+(1-p_1)D\rho D^*,
\end{equation}
where $D=\diag\{1, e^{i\theta}\}$ for $\theta\in \mathcal{R}$, $0<p<1$ and $\rho\in D(H)$.\\
\\
According to Theorem \ref{2d}, superpositions $|\varphi_1\rangle$ and $|\varphi_2\rangle$ are not preserved by $\Phi_0$ if  they are distinguishable. Taking $|\varphi_1\rangle$, $|\varphi_2\rangle$ as non-distinguishable states, Theorem \ref{18} guarantees that the fidelity can be preserved if $|\varphi_1\rangle $ and $|\varphi_2\rangle $ are $D-$symmetric and $|R_D(|\varphi_1\rangle, |\varphi_2\rangle)|=|<\varphi_1, \varphi_2>|$.\\
\\
By removing the relative phases of $|\varphi_1\rangle$ and $|\varphi_2\rangle$, we take $|\varphi_1\rangle=\begin{pmatrix} r\\ \sqrt{1-r^2}e^{i\theta_1}\ \end{pmatrix}$, $|\varphi_2\rangle=\begin{pmatrix} s\\ \sqrt{1-s^2}e^{i\theta_2}\ \end{pmatrix}$ for nonzero $r, s \in \mathcal{R}$ and $\theta_1. \theta_2\in \mathcal{R}$. Then  $|\varphi_1\rangle $ and $|\varphi_2\rangle $ are $D-$symmetric if and only if 
\begin{equation}
	\frac{(r, \sqrt{1-r^2}e^{-i\theta_1})D\begin{pmatrix} s\\ \sqrt{1-s^2}e^{i\theta_2}\ \end{pmatrix}}{rs+\sqrt{1-r^2}\sqrt{1-s^2}e^{i(\theta_2-\theta_1)}}=\frac{(s, \sqrt{1-s^2}e^{-i\theta_2})D\begin{pmatrix} r\\ \sqrt{1-r^2}e^{i\theta_1}\ \end{pmatrix}}{rs+\sqrt{1-r^2}\sqrt{1-s^2}e^{i(\theta_1-\theta_2)}}
\end{equation}
which implies,
\begin{align}
	&rs\sqrt{1-r^2}\sqrt{1-s^2}e^{i(\theta_2-\theta_1)}-rs\sqrt{1-r^2}\sqrt{1-s^2}e^{i(\theta_1-\theta_2)}=\\
	&(rs\sqrt{1-r^2}\sqrt{1-s^2}e^{i(\theta_2-\theta_1)}-rs\sqrt{1-r^2}\sqrt{1-s^2}e^{i(\theta_1-\theta_2)})e^{i\theta}
\end{align}	
Then we have $\theta=0$ (contradicts with $Cr(\Phi)=2$) or $\theta_1-\theta_2=k\pi$ for $k=0$ or $1$.\\

\noindent Further, the condition $|R_D(|\varphi_1\rangle, |\varphi_2\rangle)|=|<\varphi_1, \varphi_2>|$ is equivalent to
\begin{equation}\label{4.8}
	|rs+\sqrt{1-r^2}\sqrt{1-s^2}e^{i(\theta+k\pi)}|=|rs+\sqrt{1-r^2}\sqrt{1-s^2}e^{ik\pi}|.
\end{equation}	
By squaring both sides of Equation \ref{4.8}, we have
	\begin{align}
		rs\sqrt{1-r^2}\sqrt{1-s^2}e^{-i(\theta+k\pi)}&+rs\sqrt{1-r^2}\sqrt{1-s^2}e^{i(\theta+k\pi)}\\
		=rs\sqrt{1-r^2}\sqrt{1-s^2}e^{-ik\pi}&+rs\sqrt{1-r^2}\sqrt{1-s^2}e^{ik\pi},
	\end{align}
which shows that
\begin{equation}
		e^{-i(\theta+k\pi)}+e^{i(\theta+k\pi)}=e^{-ik\pi}+e^{ik\pi},
\end{equation}
and then $\theta=0$ (contradiction).\\

\noindent We conclude that such a superposition pair $ |\varphi_1\rangle $, $ |\varphi_2\rangle $ cannot exist in $S$, and the size of $S$ is at most 3, occurring only when $ E_2 \subset S $.
\end{proof}

\begin{proof}[Relation between Phase Damping and General Phase Damping channels]
We first show that the Phase damping channel $\Phi_P$ is a general phase damping channel.\\
\\
 Consider a rank-2 general phase damping channel $\Phi(\rho)=p_1\rho+(1-p_1)D_0\rho D_0^*$ with $0\leq p_1\leq 1$ and $D_0=\begin{pmatrix}{}
  1&0  \\
  0&e^{i\theta} 
\end{pmatrix}
$ for some $\theta\in \mathcal{R}$. Then by Theorem \ref{Choi-thm2}, it is sufficient to show that the Kraus operators $\left\{ K_1, K_2\right\}$ and $\left\{\sqrt{p_1}I, \sqrt{1-p_1}D_0\right\}$ are equivalent under some unitary matrix 
\begin{equation}
	U=\begin{pmatrix}{}
  u_1&u_2  \\
  u_3&u_4 
\end{pmatrix}.
\end{equation}
By taking $\theta=\pi$, $p_1=\frac{1+\sqrt{1-\lambda}}{2}$ and 
\begin{equation}
	U=\begin{pmatrix}{}
  \sqrt{p_1}& \sqrt{\frac{p_1}{\lambda}}(1-\sqrt{1-\lambda})\\
  \sqrt{1-p_1}&  \sqrt{\frac{1-p_1}{\lambda}}(e^{i\theta}-\sqrt{1-\lambda})
\end{pmatrix},
\end{equation}
we have 
\begin{equation}
	\Phi_P(\rho)=\frac{1+\sqrt{1-\lambda}}{2}\rho+\frac{1-\sqrt{1-\lambda}}{2}\begin{pmatrix}{}
  1&0  \\
  0&-1 
\end{pmatrix}\rho
\begin{pmatrix}{}
  1&0  \\
  0&-1 
\end{pmatrix},
\end{equation}
which is in the form of a general phase damping channel for all $\rho\in D(H)$ and $p_i\neq 0, i=1,2$.\\
\\
We can also easily verify by the same way that the well known quantum errors introduced by Amplitude Damping Channel $\Phi_A$ and Depolarizing Channel $\Phi_D$ are not General Phase Damping channel if they are not degenerated. For any $\rho\in D(H)$, the Amplitude damping channel $\Phi_A$ with amplitude parameter $0\leq \eta \leq 1$ is usually defined as $\Phi_A(\rho)=S_1\rho S_1^*+S_2\rho S_2^*$ with
\begin{equation}
		S_1=\begin{pmatrix}{}
  1&0 \\
  0&\sqrt{1-\eta} 
\end{pmatrix}, 
S_2=\begin{pmatrix}{}
  0&\sqrt{\eta} \\
  0&0 
\end{pmatrix}
	\end{equation}
and the depolarizing channel $\Phi_D$ with depolarizing parameter $0\leq \zeta \leq \frac{4}{3}$ is defined as $\Phi_D(\rho)=(1-\zeta)\rho+\frac{\zeta}{2}I$ with Kraus operators
\begin{equation}
	T_1=\sqrt{1-\frac{3\zeta}{4}}I, T_2=\sqrt{\frac{\zeta}{4}}\sigma_x, T_3=\sqrt{\frac{\zeta}{4}}\sigma_y, T_4=\sqrt{\frac{\zeta}{4}}\sigma_z. 
\end{equation}  
\\

\end{proof}

\section*{Statements and Declarations}

\noindent \textbf{Competing Interests.}
The authors declare no competing interests.\\
\\
\noindent \textbf{Author Contributions.}
Both authors contributed equally to the conception of the study, the analysis, and the writing of the manuscript. Both authors reviewed and approved the final version of the manuscript.\\
\\
\noindent \textbf{Data availability statement:}
No new data were created or analyzed in this study.


\end{document}